\documentclass[11pt,a4paper]{article}

\usepackage{amssymb,amsthm,amsmath,mathtools}
\usepackage{thmtools,thm-restate}
\usepackage[round]{natbib}
\usepackage{pifont}

\usepackage[usenames,svgnames,xcdraw,table]{xcolor}
\definecolor{DarkGreen}{rgb}{0.1,0.5,0.1}
\usepackage[backref=page]{hyperref}
\hypersetup{
	colorlinks=true,
	linkcolor=red,
	urlcolor=DarkGreen,
	citecolor=blue
}
\renewcommand*{\backref}[1]{}
\renewcommand*{\backrefalt}[4]{%
    \ifcase #1 (Not cited.)%
    \or        (Cited on page~#2)%
    \else      (Cited on pages~#2)%
    \fi}
\usepackage[capitalise,noabbrev]{cleveref}
\Crefname{property}{Property}{Properties}
\Crefname{example}{Example}{Examples}
\Crefname{table}{Table}{Tables}

\usepackage{cancel}
\usepackage{nicefrac}
\usepackage{enumerate}
\usepackage{etoolbox}
\usepackage[margin=0.99in]{geometry}
\usepackage[T1]{fontenc}
\usepackage[tt=false]{libertine}
\usepackage[spacing=true,factor=1100,stretch=10,shrink=10]{microtype}
\microtypecontext{spacing=nonfrench}

\usepackage{url}
\usepackage[utf8]{inputenc}
\usepackage[small]{caption}
\usepackage{graphicx}
\usepackage{booktabs}
\usepackage{mathrsfs,amsfonts,dsfont,authblk}
\usepackage{array,multirow,graphicx,bigdelim}

\usepackage[linesnumbered,lined,boxed,ruled,vlined]{algorithm2e}

\SetKwInOut{Parameters}{Parameters}
\SetKwComment{Comment}{$\triangleright$\ }{}
\SetAlFnt{\small}
\SetAlCapFnt{\small}
\SetAlCapNameFnt{\small}
\SetAlCapHSkip{0pt}
\IncMargin{-\parindent}

\usepackage{tikz}
\usepackage{tikz-cd}
\usetikzlibrary{fit,calc,shapes,decorations.text,arrows,decorations.markings,decorations.pathmorphing,shapes.geometric,positioning,decorations.pathreplacing}
\tikzset{snake it/.style={decorate, decoration=snake}}

\colorlet{mygray}{gray!40}

\makeatletter
\let\oldnl\nl
\newcommand{\nonl}{\renewcommand{\nl}{\let\nl\oldnl}}
\makeatother

  {\list{}{\leftmargin=#1\rightmargin=#1}\item[]}%
  {\endlist}

\usepackage[labelfont={normalfont,bf},textfont=it]{caption}
\usepackage{subcaption}

\newtheorem{theorem}{Theorem}
\newtheorem{lemma}{Lemma}
\newtheorem{corollary}{Corollary}
\theoremstyle{definition}

\theoremstyle{remark}

\usepackage{flushend}
\usepackage[utf8]{inputenc}
\usepackage[inline]{enumitem}
\Crefname{claim}{Claim}{Claims}

\newtheorem{claim}{Claim}

\allowdisplaybreaks

\usepackage[compact]{titlesec}

\usepackage{amssymb,amsthm,amsmath,mathtools}

\usepackage[labelfont={normalfont,bf},textfont=it]{caption}
\usepackage{subcaption}

\newtheorem{observation}{Observation}

\usepackage[title]{appendix} 

\newcommand{\NPH}{\textrm{\textup{NP-hard}}}
\newcommand{\NPC}{\textrm{\textup{NP-complete}}}

\newcommand{\C}{\mathcal{C}}

\title{Multiwinner Elections under Minimax Chamberlin-Courant Rule in Euclidean Space}

\author[1]{Chinmay Sonar, Subhash Suri}
\author[2]{Jie Xue}
\affil[1]{University of California, Santa Barbara, USA\\
	{\small\texttt{\{csonar, suri\}@cs.ucsb.edu}}}
\affil[2]{New York University, Shaghai, China\\
	{\small\texttt{jiexue@nyu.edu}}}

\date{}

\begin{document}

\maketitle

\begin{abstract}
We consider multiwinner elections in Euclidean space using the minimax Chamberlin-Courant rule.
In this setting, voters and candidates are embedded in a $d$-dimensional Euclidean space,
and the goal is to choose a committee of $k$ candidates so that the \emph{rank} of any voter's
most preferred candidate in the committee is minimized. (The problem is also equivalent to the 
\emph{ordinal} version of the classical $k$-center problem.) 
We show that the problem is \NPH{} in any dimension $d \geq 2$, and also provably hard to approximate.
Our main results are three polynomial-time approximation schemes, each of which finds a committee 
with provably good minimax score. In all cases, we show that our approximation bounds are tight or close to tight.
We mainly focus on the $1$-Borda rule but some of our results also hold for the more general $r$-Borda.
\end{abstract}

\section{Introduction}
Multiwinner elections are a classical problem in social science where the goal is to choose a fixed number of candidates, also called a \emph{winning committee}, based on voters' preferences. The problem encapsulates a number of applications that range from choosing representatives in democracies to staff hiring and procurement decisions~\cite{lu2011budgeted,goel2019knapsack,skowron2016finding}. 
Many facility location problems in operations research and public goods planning are equivalent to multiwinner elections: facilities are candidates and users are voters~\cite{betzler2013computation}.

One of the central computational problems in this area is to design algorithms for computing a winning committee with \emph{provable} guarantees of representation quality. In particular, let $V$ be a set of $n$ voters, $\mathcal{C}$ 
a set of $m$ candidates, and $k$ the size of the winning committee, for a positive integer $k$. 
We focus on \emph{ordinal} preferences where each voter $v$ ranks all the candidates in $\mathcal{C}$, from the most preferred (rank $1$) to the least preferred (rank $m$). 
Elections under ordinal preferences are widely studied because in many settings, the rank ordering of candidates 
is both natural and easier to determine than a precise numerical value (cardinal preference)~\cite{brandt2016handbook,endriss2017trends,munagala2021optimal}.
We let $\sigma_v (c)$ denote the rank of candidate $c$ in $v$'s list, and use Chamberlin-Courant voting rule~\cite{chamberlin1983representative} for evaluating the score of a committee. 
In the simplest version, also called $1$-Borda, a voter $v$'s score for a committee $T$ is the rank of its most preferred committee member:
$\sigma_v(T) ~=~ \min_{c \in T} \sigma_v(c)$.
(In a generalized form called $r$-Borda, which we also consider, the score is the \emph{sum} of the ranks of the $r$ most preferred candidates in $T$.)
This scoring function is often called the \emph{misrepresentation} score to emphasize that we want to \emph{minimize} it---more preferred candidates have smaller ranks.

The goal of a multiwinner election is to find a size-$k$ committee $T \subseteq \mathcal{C}$ that optimizes some function $g(\sigma_{v_1} (T), \ldots, \sigma_{v_n} (T))$ of all the voter's scores.  Two classical choices for $g$ are the \emph{sum} and the \emph{max}. 
The former is the \emph{utilitarian} objective and seeks to minimize the \emph{sum} of the scores over all the voters. 
The latter is the \emph{egalitarian} objective and minimizes the \emph{maximum} (worst) of the scores over all the voters.  
Both versions of the multiwinner elections are \NPH{} under \emph{general preferences}~\cite{lu2011budgeted,betzler2013computation}, and as a result, an important line of research has been to examine natural settings with \emph{structured} preference spaces~\cite{betzler2013computation,yu2013multiwinner,skowron2015complexity,elkind2017structured}.

Our paper studies one such setting---and arguably one of the most natural---namely, the Euclidean space of preferences.
The geometry of Euclidean space gives an intuitive and interpretable \emph{positioning} of voters and candidates in many natural settings such as spatial voting and facility locations, but it also has important computational  advantages: when candidates and voters are embedded in $d$-space, only a tiny fraction of all (exponentially many) $m!$ candidate orderings are \emph{realizable}. In particular, the maximum number of realizable rankings is only (polynomially bounded) $O(m^{d+1})$. This important combinatorial property enables us derive much better bounds than what is possible in completely unstructured preferences spaces.
Specifically, we explore algorithmic and hardness questions for \emph{minimax} (egalitarian) multiwinner elections using the $1$-Borda and $r$-Borda rules under ordinal \emph{Euclidean} preferences. In this setting, voters and candidates are embedded in a $d$-dimensional Euclidean space, implicitly specifying each voter's ranking (closest to farthest) of the candidates. The goal is to choose a committee of $k$ candidates minimizing the rank of the worst voter's most preferred candidate. 

\medskip

\noindent
{\bf Remark:}
There are good reasons for using \emph{ordinal} preferences even when cardinal distances are implied by an Euclidean embedding.
The first is robustness: consider a voter $v$ and two candidates $c, c'$. If their distances satisfy $d(v,c) < d(v, c')$, then clearly $v$ prefers $c$ to $c'$, but it seems harder to argue that $v$'s preference varies linearly (or even smoothly) with  
distance---for instance, would doubling the distance really halve the value to a voter?
Another reason is that $k$-center solutions based on cardinal preferences are highly susceptible to the outlier
effect---a few outlying voters may control the minimax value (i.e., radius) of the optimal solution even though
all other voters have significantly better solution quality. 
By contrast, under the ordinal measure the (rank-based) solution seems more equitable because outliers are matched 
with a highly ranked candidate (irrespective of the distances).

\subsection{Our Results}

Our work shows that a number of interesting and encouraging approximation results are possible for Euclidean preferences, thus, offering new directions for research. Indeed, quoting~\cite{elkind2017structured}, ``multidimensional domain restrictions offer many challenging research questions, but fast algorithms for these classes are very desirable.'' 
A brief summary of our main results is the following:

\begin{enumerate}

\item We show that the Euclidean minimax committee problem is \NPH{} in every dimension $d \geq 2$; in one dimension, the problem is easy to solve optimally with dynamic programming. The complexity of this problem, also called the ordinal Euclidean $k$-center
problem, was not known and had been an important folklore problem. Our proof shows that the problem also hard to 
approximate in the worst-case (see Theorem \ref{thm-hardtoappx}), which stands in sharp contrast with the $2$-approximability of the \emph{cardinal} $k$-center problem~\cite{kleinberg2006algorithm}. 

\item We then show a number of efficient approximation results for the problem, starting with a polynomial time algorithm
to compute a size-$k$ committee with a minimax score of $O(m/k)$ for any instance in dimension $d=2$, and 
score of $O( (m/k) \log k)$ for any instance in dimension $d \geq 3$. These scores are
also shown to be essentially the best possible in worst-case. 

\item Our next approximation uses the bicriterion framework to design a polynomial time algorithm that achieves the 
\emph{optimal} minimax score $\sigma^\star$ possible for a size-$k$ committee by constructing slightly larger committee, 
namely, of size $(1 + \epsilon) k$ for $d=2$ and size $O(k \log m)$ for $d \geq 3$.
(We also show that increasing the committee size by an \emph{additive} constant is not sufficient.)

\item Our final approximation combines ordinal and cardinal features of the problem in a novel way, as follows.
Suppose the optimal score of the $k$ committee is $\sigma^\star$, and $d_v^\star$ is the distance of $v$ to its 
rank $\sigma^\star$ candidate. We define a committee $T$ to be \emph{$\delta$-optimal} if each voter has a 
representative in $T$ within distance $\delta d_v^\star$. (That is, for each voter the committee contains a
candidate whose distance to the voter is almost as good as distance to its $\sigma^*$ rank candidate.)
We show that a $\delta$-optimal committee can be computed in polynomial time for $\delta = 3$, 
but unless P = NP, there is no polynomial-time algorithm to compute $\delta$-optimal committees for any $\delta < 2$.

\end{enumerate}

\section{Related Work}
The literature on multi-winner elections is too large to summarize in this limited space; hence, we mostly survey the computational results closely related to our work.
For a general introduction to multi-winner elections, we refer the reader to the works of~\cite{elkind2017properties,faliszewski2017multiwinner,faliszewski2019committee}. 
The work of \cite{ijcai2018-8} studies the computational complexity and axiomatic properties of various egalitarian committee scoring rules under the general preferences.
For computing a winning committee under the Chamberlin-Courant rule, polynomial-time algorithms are known only for restricted preferences such as single-peaked, single-crossing, 1D Euclidean,
etc.~\cite{betzler2013computation,skowron2015complexity,elkind2017structured}.
Very little is known about the more general $d$-dimensional Euclidean setting considered in our paper
with the exception of a work of~\cite{godziszewski2021analysis}, which shows \NPH{}ness for the
\emph{approval set} voting rule for the utilitarian objective in $2$-dimensional Euclidean elections.

Constant factor approximations are often easier to achieve under the \emph{utilitarian} objective.
For instance,~\cite{munagala2021optimal} present several nearly-optimal approximation bounds;
~\cite{skowron2015achieving} presents a constant factor approximation for minimizing the weighted sum of 
ranks of the winning candidates; and~\cite{byrka2018proportional} presents a constant factor approximation for
minimizing the sum when $\sigma_v(c)$ is an arbitrary cardinal value.
In contrast, for the minimax objective, mostly inapproximability results are known, and only under general preferences~\cite{skowron2015achieving}.

\section{Preliminaries}
\label{sec:prelims}

Throughout the paper, an \textit{election} is a pair $E=(\mathcal{C}, V)$, where 
$\mathcal{C}=\{c_1, \ldots, c_m\}$ is the \textit{candidate} set and $V = \{v_1, \ldots, v_n\}$ is the \textit{voter} set.
The \textit{preference list} of each voter is a total ordering (ranking) of $\mathcal{C}$, in which 
the most preferred candidate has rank $1$ and the least preferred candidate has rank $m$.

We call an election $E = (\mathcal{C},V)$ a $d$-\textit{Euclidean election} if there exists a function 
$f: \mathcal{C} \cup V \to \mathbb{R}^d$, called a \textit{Euclidean realization} of $E$, such that for any pair $c_i,c_j \in \C$, a voter $v$ prefers $c_i$  to $c_j$ if and only if 
$\text{dist}(f(v),f(c_i)) < \text{dist}(f(v),f(c_j))$ where $\text{dist}(\cdot,\cdot)$ denotes the Euclidean distance (in case the candidates are equidistant, we break the ties arbitrarily).
We assume that a Euclidean realization is part of the input; the decision problem of whether an
election admits an Euclidean realization is computationally hard \cite{peters2017euclidean}.


We use $\sigma_v (c)$ to denote the rank of candidate $c$ in $v$'s preference list, and use the Chamberlin-Courant voting rule~\cite{chamberlin1983representative} for evaluating the score of a committee. In particular, a voter $v$'s 
score for a committee $T$ is the rank of its most preferred committee member $\sigma_v(T) ~=~ \min_{c \in T} \sigma_v(c)$,
and the \emph{Euclidean minimax committee} problem is to choose a committee of size $k$ that minimizes the
maximum score (misrepresentation) of any voter. That is, minimize the following:
$$ \sigma(T) = \max\limits_{v \in V} \left(\min\limits_{c \in T} \sigma_v(c) \right),$$
So the optimal committee score is always between $1$ and $m$.
(We mainly focus on this simpler $1$-Borda scoring function, but some of our results also hold
for the generalized $r$-Borda, as mentioned later.)

\section{Hardness Results}
\label{sec:hardness-results}

We first show that the Euclidean minimax committee problem is \NPH{} in any dimension $d \geq 2$.
After that, we extend our proof to show that the problem is even hard to approximate within any sublinear factor of $m$, where $m$ is the number of candidates in the election.

Our hardness reduction uses the \NPC{} problem \textsc{Planar Monotone 3-SAT} (PM-3SAT)~\cite{deBerg2010SAT}.
An instance of PM-3SAT consists of a \emph{monotone} 3-CNF formula $\phi$ where each clause contains either three positive literals or three negative literals, and a special ``planar embedding'' of the variable-clause incidence graph of $\phi$ described below.
In the embedding, each variable/clause is drawn as a (axis-parallel) rectangle in the plane.
The rectangles for the variables are drawn along the $x$-axis, while the rectangles for positive (resp., negative) clauses lie above (resp., below) the $x$-axis.
All the rectangles are pairwise disjoint.
If a clause contains a variable, then there is a vertical segment connecting the clause rectangle and the variable rectangle.
Each such vertical segment is disjoint from all the rectangles except the two it connects.
We call such an embedding a \textit{rectangular embedding} of $\phi$.
See Figure~\ref{fig:PM-3SAT-instance} for an illustration. 
Given $\phi$ with a rectangular embedding, the goal of PM-3SAT is to decide if there exists a satisfying assignment for $\phi$.

\begin{figure}
\centering
\begin{subfigure}[b]{0.5\textwidth}
\centering
  \includegraphics[scale = 0.4]{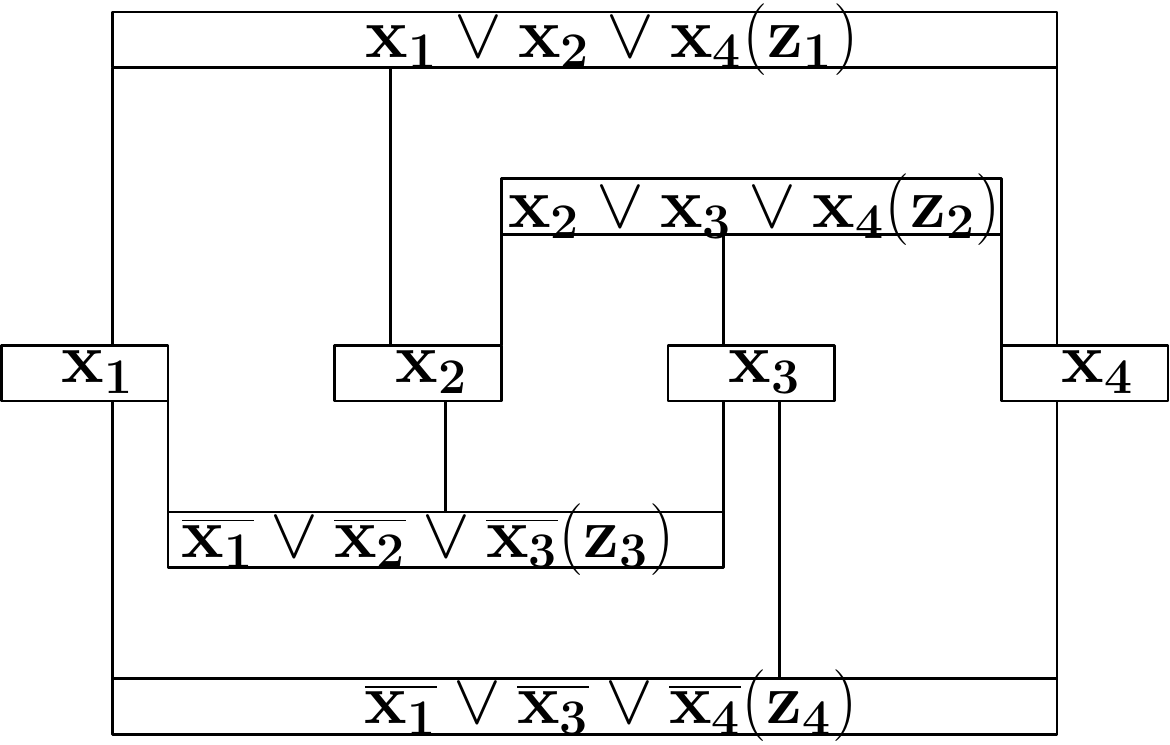}
  \caption{Rectangular embedding}
  \label{fig:PM-3SAT-instance}
\end{subfigure}
\begin{subfigure}[b]{0.4\textwidth}
    \centering
  \includegraphics[scale = 0.4]{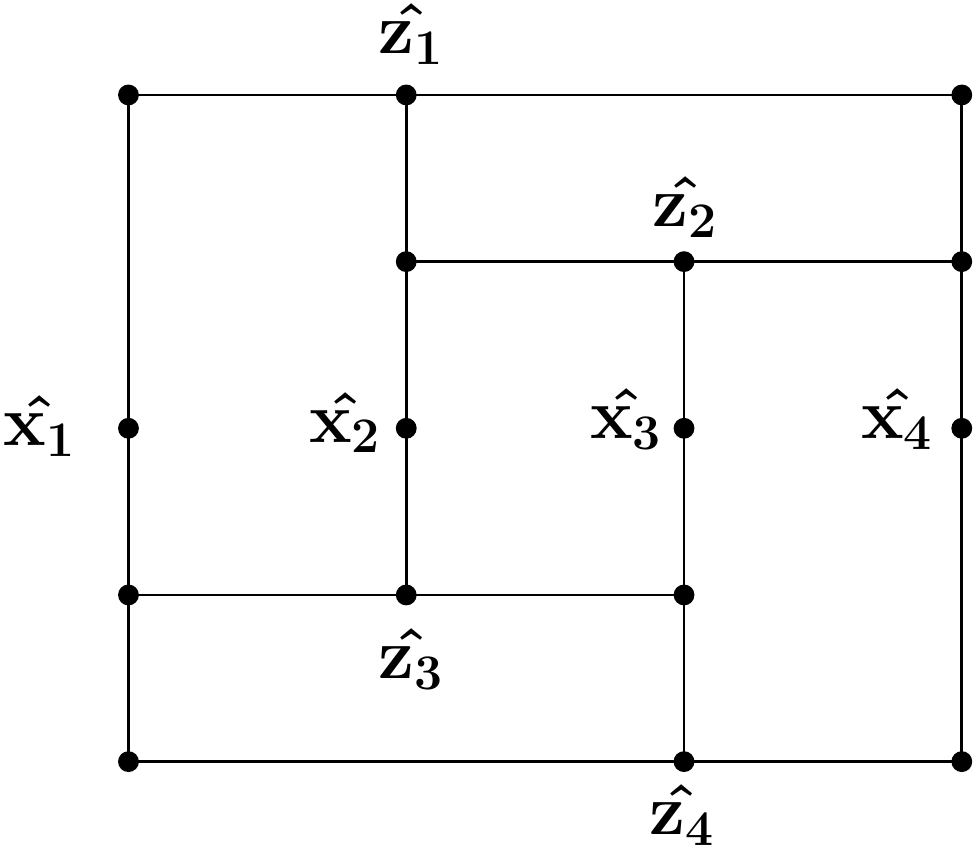}
  \caption{Orthogonal embedding}
  \label{fig:PM-3SAT-orthogonal-embedding}
\end{subfigure}
\caption{Embeddings of an example PM-3SAT instance}
\end{figure}


Suppose we are given a PM-3SAT instance consisting of the monotone 3-CNF formula $\phi$ with a rectangular embedding.
Let $X = \{x_1,\dots,x_n\}$ be the set of variables and $Z = \{z_1,\dots,z_m\}$ be the set of clauses of $\phi$ (each of which consists of three literals).
We  construct (in polynomial time) a Euclidean minimax committee instance $(E,k)$ in $\mathbb{R}^2$ such that $\phi$ has a satisfying assignment if and only if for the election $E$ there exists a committee of size $k$ with score at most $4$.
We begin by modifying the rectangular embedding of $\phi$ to another form that we call an \textit{orthogonal embedding}.

\paragraph{Orthogonal Embedding}
First, we collapse all the variable and clause rectangles into horizontal segments: the segments for variables lie on the $x$-axis and those for positive (resp., negative) clauses lie above (resp., below) the $x$-axis.
By slightly moving the clause segments vertically, we can guarantee that all clause segments have distinct $y$-coordinates.
Consider a variable/clause segment $s$.
There are vertical segments connected to $s$; we call the intersection points of these segments with $s$ the \textit{connection points}.
We then only keep the part of $s$ that is in between the leftmost and rightmost connection points; that is, we truncate the part of $s$ that is to the left (resp., right) of the leftmost (resp., rightmost) connection point.
Next, we shrink each variable segment $s$ to a point $\hat{s}$ and also move the vertical segments incident to $s$ accordingly.
By doing this, all vertical segments incident to $s$ are merged into a single vertical segment going through $\hat{s}$, whose two endpoints lie on the highest positive $s$-neighboring clause segment and the lowest negative $s$-neighboring clause segment; furthermore, this segment hits one endpoint of each of the other $s$-neighboring clause segments.
After shrinking all variable segments, we obtain our orthogonal embedding for $\phi$.
Figure~1b shows an illustration of the orthogonal embedding.
In the orthogonal embedding, each variable corresponds to a point on the $x$-axis (which we call the \textit{reference point} of the variable) and a vertical segment through the reference point, while each positive (resp., negative) clause corresponds to a horizontal segment above (resp., below) the $x$-axis; we call the intersection points of these vertical and horizontal segments \textit{connection points}.

The resulting embedding 
satisfies the following three conditions:
\textit{(i)} no vertical segment crosses a horizontal segment;
\textit{(ii)} each horizontal segment $s$ intersects exactly three vertical segments which correspond to the three variables contained in the clause corresponding to $s$; and
\textit{(iii)} the endpoints of all segments are \textit{connection points}.

By properties \textit{(ii)} and \textit{(iii)}, there are three connection points on each clause segment, two of which are the left and right endpoints of the segment, and we call the middle one the \textit{reference point} of the clause.
We denote by $\hat{x}_1,\dots,\hat{x}_n$ the reference points of the variables $x_1,\dots,x_n$ and denote by $\hat{z}_1,\dots,\hat{z}_m$ the reference points of the clauses $z_1,\dots,z_m$.
By shifting/scaling the segments properly without changing the topological structure of the orthogonal embedding, we can further guarantee that the $x$-coordinates (resp., $y$-coordinates) of the vertical (resp., horizontal) segments are distinct \textit{even} integers in the range $\{1,\dots,2n\}$ (resp., $\{-2m,\dots,2m\}$).
Therefore, all the connection points now have integral coordinates (which are even numbers) and the entire embedding is contained in the rectangle $[1,2n] \times [-2m,2m]$.
Points on the segments of the orthogonal embedding that have integral coordinates partition each segment $s$ into $\ell(s)$ unit-length segments, where $\ell(s)$ is the length of $s$.
We call these unit-length segments the \textit{pieces} of the orthogonal embedding.
Let $N$ be the total number of pieces.
Clearly, $N = O(nm)$.

Our Euclidean minimax committee instance $(E,k)$ consists of the following set of
voters and candidates in the two-dimensional plane.
(In fact, in our construction, each point is both a candidate and a voter, namely, 
$\mathcal{C} = V$. It is easy to modify the construction so that the set of voters is much larger by
simply making multiple copies of each voter.)

\paragraph{Variable gadgets.}
For each variable $x_i$, we choose four points near the reference point $\hat{x}_i$ as follows.
There are two (vertical) pieces incident to $\hat{x}_i$ in the orthogonal embedding, one above $\hat{x}_i$, the other below $\hat{x}_i$.
On each of the two pieces, we choose two points with distances $0.01$ and $0.02$ from $\hat{x}_i$, respectively.
We put a candidate and a voter at each of the four chosen points, and call these candidates/voters the $x_i$-\textit{gadget}.
We construct gadgets for all $x_1,\dots,x_n$.
The total number of candidates/voters in the variable gadgets is $4n$.

\paragraph{Clause gadgets.}
The second set of candidates/voters are constructed for the clauses $z_1,\dots,z_m$.
For each clause $z_i$, we put a candidate and a voter at the reference point $\hat{z}_i$, and call this candidate/voter the $z_i$-\textit{gadget}.
The total number of candidates/voters in the clause gadgets is $m$.

\paragraph{Piece gadgets.}
The last set of candidates/voters are constructed for connecting the variable gadgets and the clause gadgets.
Consider a piece $s$ of the orthogonal embedding, which is a unit-length segment.
We distinguish the two endpoints of $s$ as $s^-$ and $s^+$ as follows.
If $s$ is a vertical piece above (resp., below) the $x$-axis, let $s^-$ be the bottom (resp., top) endpoint of $s$ and $s^+$ be the top (resp., bottom) endpoint of $s$.
If $s$ is a horizontal piece, then it must belong to the horizontal segment of some clause $z_i$.
If $s$ is to the left (resp., right) of the reference point $\hat{z}_i$, let $s^-$ be the left (resp., right) endpoint of $s$ and $s^+$ be the right (resp., left) endpoint of $s$.
For every piece $s$ that is \textit{not} adjacent to any clause reference point, we choose four points on $s$ with distances $0.49$, $0.8$, $0.9$, $1$ from $s^-$ (i.e., with distances $0.51$, $0.2$, $0.1$, $0$ from $s^+$), respectively.
We put a candidate and a voter at each of the four chosen points, and call these the candidates/voters of the $s$-\textit{gadget}.
Note that we do not construct gadgets for the pieces that are adjacent to some clause reference point.
Thus, the total number of candidates/voters in the piece gadgets is $4(N-3m)$, as each clause reference point is adjacent to three pieces.


By combining these three constructed gadgets, we obtain our election $E = (\mathcal{C},V)$ instance with $4N+4n-11m$ candidates and voters. The size of the committee is $k = N+n-3m$.
We now prove that $E$ has a committee of size $k$ with score $\leq 4$ iff $\phi$ is satisfiable.

\medskip

\paragraph{The ``if'' part.~~~}
Suppose $\phi$ is satisfiable and let $\pi: X \rightarrow \{\mathsf{true},\mathsf{false}\}$ be an assignment which makes $\phi$ true.
We construct a committee $T \subseteq \mathcal{C}$ of size $k$ as follows.
Our committee $T$ contains one candidate in each variable gadget and each piece gadget (this guarantees $|T| = k$ as the total number of variable and piece gadgets is $k$).
Consider a variable $x_i$.
By our construction, the $x_i$-gadget contains four candidates which have the same $x$-coordinates as $\hat{x}_i$.
If $\pi(x_i) = \mathsf{true}$ (resp., $\pi(x_i) = \mathsf{false}$), we include in $T$ the topmost (resp., bottommost) candidate in the $x_i$-gadget.
Now consider a piece $s$ that is not adjacent to any clause reference point.
We first determine a variable as the \textit{associated} variable of $s$ as follows.
If $s$ is vertical, then the associated variable of $s$ is just defined as the variable whose vertical segment contains $s$.
If $s$ is horizontal, then $s$ must belong to the horizontal segment of some clause $z_j$.
In this case, we define the associated variable of $s$ as the variable whose vertical segment intersects the left (resp., right) endpoint of the horizontal segment of $z_j$ if $s$ is to the left (resp., right) of the reference point $\hat{z}_j$.
Let $x_i$ be the associated variable of $s$.
If $\pi(x_i) = \mathsf{true}$, then we include in $T$ the candidate in the $s$-gadget that has distance $1$ (resp., $0.9$) from $s^-$ if $s$ is above (resp., below) the $x$-axis.
Symmetrically, if $\pi(x_i) = \mathsf{false}$, then we include in $T$ the candidate in the $s$-gadget that has distance $1$ (resp., $0.9$) from $s^-$ if $s$ is below (resp., above) the $x$-axis.
This finishes the construction of $T$.
The following lemma completes the ``if'' part of our proof.

\begin{lemma}
\label{lem:equivalence-if-part}
The score of $~T$ in the election $E$ is at most $4$.
\end{lemma}

Proof of \cref{lem:equivalence-if-part} is presented in the Appendix~\ref{app:proof-lemma-1}.

\medskip

\paragraph{The ``only if'' part.}
Suppose there exists a size-$k$ committee $T \subseteq \mathcal{C}$ with score at most $4$.
We use that committee to construct a satisfying assignment $\pi: X \rightarrow \{\mathsf{true},\mathsf{false}\}$. We first note the following property of the committee $T$.
\begin{lemma} \label{lem-onecand}
$T$ contains exactly one candidate in each variable gadget and exactly one candidate in each piece gadget. 
\end{lemma}
\begin{proof}
We first show that $T$ contains at least one candidate in each variable gadget.
For each variable $x_i$, recall that we place four voters (say $v_1, v_2, v_3, v_4$) and four candidates (say $c_1, c_2, c_3, c_4$) near the reference point $\hat{x}_i$.
Among the four candidates (resp., voters), $c_1, c_2$ (resp., $v_1, v_2$) are at distances $0.01$ and $0.02$ above $\hat{x}_i$, and $c_3, c_4$ (resp. $v_3, v_4$) are at distances $0.01$ and $0.02$ below $\hat{x}_i$, respectively.
We recall from the proof of Lemma 1 (paragraph 2) that the closest four candidates to each of $v_1, v_2, v_3$ and $v_4$ belong to the set $\{c_1, c_2, c_3, c_4\}$.
Since the score of committee $T$ is at most $4$, then, in particular, $\sigma_{v_i}(T) \leq 4$ for each $i \in [4]$; hence, $T$ contains at least one of $c_1, c_2, c_3, c_4$.
This completes the argument for the variable gadgets.
Next, consider a piece gadget $s$. 
Let $v_1, v_2, v_3$, and $v_4$ (resp., $c_1, c_2, c_3$, and $c_4$) be the voters (resp., candidates) placed at distances $0.49, 0.8, 0.9$, $1$ from $s^-$, respectively.
For the voters $v_2$ and $v_3$ a candidate from an adjacent variable/piece/clause gadget is at least at a distance $0.78/0.57/1.1$, respectively, while the candidates $c_1, c_2, c_3, c_4$ lie within distance $0.41$. 
Therefore, the closest four candidates for $v_2$ and  $v_3$ belong to the set $\{c_1, c_2, c_3, c_4\}$.
Hence, $T$ includes at least one of $c_1, c_2, c_3, c_4$ for each piece gadget.

At this stage, we recall that the number of variable gadgets is $n$, the number of piece gadgets is $N - 3m$, and the committee size is $k = N+n-3m$.
Hence, using a counting argument, we conclude that $T$ contains exactly one candidate from each variable gadget and each piece gadget.
\end{proof}

We note that the total number of variable and piece gadgets is $k$.
Since $|T| = k$, using Lemma~\ref{lem-onecand}, we conclude that $T$ has no budget to contain any candidate in the clause gadgets.
\begin{corollary} \label{cor-noclcand}
$T$ contains no candidate in the clause gadgets.
\end{corollary}

Recall that each variable gadget contains four candidates, two of which are above the $x$-axis while the other two are below the $x$-axis.
Consider a variable $x_i$.
By Lemma~\ref{lem-onecand}, $T$ contains exactly one candidate in the $x_i$-gadget.
If that candidate is above (resp., below) the $x$-axis, we set $\pi(x_i) = \mathsf{true}$ (resp., $\pi(x_i) = \mathsf{false}$).
We show that $\pi$ is a satisfying assignment of $\phi$.
It suffices to show that every positive (resp., negative) clause of $\phi$ contains at least one variable which is mapped to $\mathsf{true}$ (resp., $\mathsf{false}$) by $\pi$.
We only consider positive clauses, as the proof for negative clauses is similar.
We need the following property of $T$.
\begin{lemma} \label{lem-true}
Let $s$ be a piece above the $x$-axis that is not adjacent to any clause reference point, and suppose $x_i$ is the associated variable of $s$.
If $T$ contains the candidate in the $s$-gadget with distance $1$ from $s^-$, then $\pi(x_i) = \mathsf{true}$.
\end{lemma}
\begin{proof}
Consider a piece $s$ above the $x$-axis that is not adjacent to any clause reference point, and let $x_i$ be the variable associated with $s$.
First, consider the case when $s$ is not adjacent to the variable reference point $\hat{x}_i$.
Let $v_1, v_2, v_3$, and $v_4$ (resp., $c_1, c_2, c_3$, and $c_4$) be the voters (resp., candidates) placed at distances $0.49, 0.8,0.9$, and $1$ from $s^-$, respectively.
Moreover, let $s'$ be the piece below (resp., to the left of) $s$ when $s$ is a vertical (resp., horizontal) piece.
We assume that $c_4 \in T$ (note that $\text{dist}(s^-, c_4) = 1$).
We recall from the proof of Lemma 1 (paragraph 3) that the closest four candidates for $v_1$ are $c_1, c_2, c_3, c_4'$ where $c_4'$ is a candidate from the piece $s'$ placed at a distance 1 from $s'^-$.
Using Lemma~2, we know $T$ only includes $c_4$ from $s$.
Hence, to satisfy $\sigma_{v_1}(T) \leq 4$, $T$ must include the candidate $c_4'$.
Observe that we can repeat the above argument for all pieces below (resp., to the left of) $s$, which implies that for all pieces $s_i$ below (resp., to the left of) $s$, $T$ includes the candidate in $s_i$ placed at a distance $1$ from $s_i^-$.

Let $x_i$ be the variable associated with $s$ and $\hat{s}$ be the piece adjacent to the variable reference point $\hat{x}_i$.
Moreover, let $\hat{v}_1, \hat{v}_2, \hat{v}_3$, and $\hat{v}_4$ (resp., $\hat{c}_1, \hat{c}_2, \hat{c}_3$, and $\hat{c}_4$) be the voters (resp., candidates) placed at distances $0.49, 0.8,0.9$, and $1$ from $\hat{s}^-$, respectively.
The four closest candidates to $\hat{v}_1$ belong to the set $\{\hat{c}_1, \hat{c}_2, \hat{c}_3,c_5\}$ where $c_5$ is the candidate at a distance $0.02$ above $\hat{x}_i$ added corresponding to the $x_i$-gadget. 
Using the argument above, we know $T$ only includes $\hat{c}_4$ from the piece $\hat{s}$.
Hence, to satisfy $\sigma_{\hat{v}_1}(T) \leq 4$, $T$ must include the candidate $c_5$.
Recall that for an arbitrary variable $x_j$ for $j \in [n]$, if $T$ includes a candidate above the reference point $\hat{x}_j$, we set $x_j = \mathsf{true}$.
Since, $c_5 \in T$ and $c_5$ lies above $\hat{x}_i$, we set $x_i = \mathsf{true}$.
This completes the proof of Lemma~\ref{lem-true}.
\end{proof}


We will now use Lemma~\ref{lem-onecand}, Lemma~\ref{lem-true} and Corollary~\ref{cor-noclcand} to show that the constructed assignment $\pi$ satisfies $\phi$. 
We will show this only for positive clauses as the argument for negative clauses is similar.

Let $z_i$ be a positive clause. 
We want to show that at least one variable of $z_i$ is mapped to $\mathsf{true}$ by $\pi$.
Consider the three pieces $s_1,s_2,s_3$ adjacent to the reference point $\hat{z}_i$.
Suppose $s_1$ is to the left of $\hat{z}_i$, $s_2$ is to the right of $\hat{z}_i$, and $s_3$ is below $\hat{z}_i$.
By our construction, we have $\hat{z}_i = s_1^+ = s_2^+ = s_3^+$.
Since $\hat{z}_i$ is a connection point and all connection points have even coordinates, $s_1^-,s_2^-,s_3^-$ are not connection points.
Therefore, there exist pieces $s_4,s_5,s_6$ such that the right endpoint of $s_4$ is $s_1^-$, the left endpoint of $s_5$ is $s_2^-$, and the top endpoint of $s_6$ is $s_3^-$.
We have $s_1^- = s_4^+$, $s_2^- = s_5^+$, and $s_3^- = s_6^+$.
In the $s_4$-gadget, there is a candidate $c_4$ with distance $1$ from $s_4^-$ (and hence located at $s_1^-$).
Similarly, there is a candidate $c_5$ in the $s_5$-gadget located at $s_2^-$ and a candidate $c_6$ in the $s_6$-gadget located at $s_3^-$.
The candidates $c_4,c_5,c_6$, together with the candidate in the $z_i$-gadget (which is located at $\hat{z}_i$), are the four candidates closest to the voter at $\hat{z}_i$, because all pieces except $s_1,\dots,s_6$ have distances at least $2$ from $\hat{z}_i$ by the fact that the $x$-coordinates (resp., $y$-coordinates) of the vertical (resp., horizontal) pieces are all even numbers.
Since the score of $T$ is at most $4$, $T$ must contain at least one of these four candidates.
However, $T$ does not contain the candidate in the $z_i$-gadget by Corollary~1.
Thus, $T \cap \{c_4,c_5,c_6\} \neq \emptyset$.
By Lemma~3, this implies that at least one of the associated variables of $s_4,s_5,s_6$ is true.
Note that these associated variables are just the three variables in the clause $z_i$.
Therefore, $\pi$ makes $z_i$ true.
This completes the ``only if'' part of our proof.
As a result, we see that $\phi$ is satisfiable iff there exists a committee in $E$ of size $k$ whose score is at most $4$.

Finally, the reduction can clearly be done in polynomial time, and so we have established the
following result.
\begin{theorem}
\label{thm:MCC-NPH}
Euclidean minimax committee is \NPH{} in all dimensions $d \geq 2$. This claim holds even if the voter and candidate sets are identical.
\end{theorem}


In fact, our construction also rules out the possibility of a \textup{PTAS} as it is hard to decide in polynomial time whether the minimum score is $\leq 4$ or $\geq 5$.
By slightly modifying the proof, we can also show that even computing an approximation for Euclidean minimax committee within any sublinear factor of $|\mathcal{C}|$ is hard.

\begin{theorem} \label{thm-hardtoappx}
For any constant $\epsilon > 0$, it is \NPH{} to achieve a $|\mathcal{C}|^{1-\epsilon}$-approximation for Euclidean minimax committee in $\mathbb{R}^d$ for any $d \geq 2$.
\end{theorem}
\begin{proof}
Our proof is a slight modification of the proof of NP-hardness from Theorem~1. 
We will show that even computing an approximation for Euclidean minimax committee within a sublinear factor of $|\mathcal{C}|$ is hard.

Given a PM-3SAT instance with formula $\phi$, we construct the orthogonal embedding and create the clause gadgets as before.
Next, we slightly change the piece gadgets as follows.
For each piece $s$ that is not adjacent to a clause reference point, we choose $4$ points with distances $0.47, 0.8, 0.9, 1$ from $s^-$ (call them $s_1, s_2, s_3, s_4$), and introduce a voter and a candidate at each of these points.
Also, we create an additional set of candidates as follows.
Note that $s_4$ is an endpoint of $s$.
Thus, there can be two or three pieces adjacent to $s_4$, depending on whether $s_4$ is a connection point or not.
It follows that when $s$ is vertical (resp., horizontal), there is no piece adjacent to $s_4$ in one of the left or right (resp., top or bottom) directions.
Without loss of generality, assume $s$ is vertical and there is no piece adjacent to $s_4$ on its left.
Let $s_4'$ (resp., $s_3'$) be the point to the left of $s_4$ (resp., $s_3)$ with distance $0.19$ from $s_4$ (resp., $s_3)$.
We then place $(nm)^w$ candidates at $s_4'$ (resp., $s_3'$) for a sufficiently large constant $w$.
The candidates/voters at $s_1, s_2, s_3, s_4$ and the additional candidates at $s_3',s_4'$ form the $s$-gadget.
Finally, recall that in each $x_i$-gadget, we have candidates/voters at the four points near $\hat{x}_i$ two of which are at distances $0.01$ and $0.02$ above $\hat{x}_i$ and the other two of which are at distances $0.01$ and $0.02$ below $\hat{x}_i$.
Let $p^+$ (resp., $p^-$) be the point at distance $0.02$ above (resp., below) $\hat{x}_i$.
Let $q^+$ (resp., $q^-$) be the point to the left of $p^+$ (resp., $p^-$) with distance $0.05$ from $p^+$ (resp., $p^-$).
We place $(nm)^w$ additional candidates at $q^+$ (resp., $q^-$).
This completes the construction.
The desired committee size is again $k = N+n-3m$.

We first show the following structural lemma for the constructed instance.
\begin{lemma}
\label{lem-hard-appx-score}
If $\phi$ is satisfiable, then there exists a size $k$ committee with score at most 4;
otherwise, every size $k$ committee has score at least $(nm)^w$.
\end{lemma}

Proof of Lemma~\ref{lem-hard-appx-score} is presented in Appendix~\ref{app:proof-lemma-4}.
\smallskip

Our construction satisfies $|\mathcal{C}| = (4N+4n-11m) + (nm)^w N$, and since $N = O(nm)$ and we can choose any sufficiently large value for $w$, we can guarantee that $(nm)^w \geq |\mathcal{C}|^{1-\epsilon}$ for any small constant $\epsilon > 0$.
This completes the proof of Theorem~2.
\end{proof}

In the rest of this paper, we complement the hardness results of the previous sections with nearly-optimal approximation algorithms.

\section[Appx-using-eps-nets]{Approximation using $\epsilon$-nets}

Our first algorithm computes in polynomial-time a size-$k$ committee of 
minimax score $O(m/k)$ for $d=2$ and $O((m/k) \log k)$ for $d \geq 3$.
Our algorithm uses the notion of $\epsilon$-\textit{nets}, which are commonly used in computational geometry~\cite{toth2017handbook} for solving set cover and hitting set problems.
Let us first briefly describe this notion.

Let $X$ be a finite set of points in $\mathbb{R}^d$ and let $\mathcal{R}$ be a set of ranges (subsets of $X$) in $\mathbb{R}^d$.
A subset $A \subseteq X$ is called an $\epsilon$-\textit{net} of $(X,\mathcal{R})$ if $A$ intersects all those ranges in $\mathcal{R}$ that are $\epsilon$-heavy, i.e., they contain at least an $\epsilon$-fraction of the points in $X$. 
In other words, $A$ is an $\epsilon$-net for $(X,\mathcal{R})$ if $A \cap R \neq \emptyset$ for any $R \in \mathcal{R}$ with $|R \cap X| \geq \epsilon |X|$.
There exists an $\epsilon$-net of size $O(\frac{1}{\epsilon})$ for ranges defined by disks in 
$\mathbb{R}^2$, and of size $O(\frac{1}{\epsilon} \log \frac{1}{\epsilon})$ for ranges defined by balls in 
$\mathbb{R}^d$, for any constant dimension $d \geq 3$~\cite{toth2017handbook}.
In both cases, $\epsilon$-nets can be computed in polynomial time.

Building on this result, we now present our algorithm.

\begin{theorem}
\label{thm:MCC-upperbound}
Given a $d$-Euclidean election $E = (\mathcal{C},V)$, we can compute in polynomial time a size 
$k$ committee with minimax score $O(m/k)$ for $d=2$ and score $O((m/k)\log k)$ for $d \geq 3$, where $m = |\mathcal{C}|$.
\end{theorem}
\begin{proof}
In order to convey the intuition more clearly, let us first show how to find an $O(k)$-size committee with score at most $\lceil (m / k) \log k \rceil$.
Given a $d$-Euclidean election, let $\mathcal{C}$ be the set of the $m$ candidates
with their embedding in $\mathbb{R}^d$.
For each voter $v$, we consider a $d$-dimensional ball $R_v$ centered at $v$ containing the
$\lceil (m/k) \log k \rceil$ closest points of $\mathcal{C}$ to $v$. 
Let $\mathcal{R}$ be the set of all these balls.
Each ball of $\mathcal{R}$ is $\epsilon$-heavy for $\epsilon = \log k/k$ because it contains an $\epsilon$-fraction of the $m$ candidates.
Therefore, in polynomial time we can find an $\epsilon$-net $T \subseteq \mathcal{C}$ for $(\mathcal{C},\mathcal{R})$ of size $O(\frac{1}{\epsilon} \log \frac{1}{\epsilon}) = O(k)$.
By the definition of $\epsilon$-net, $T$ contains at least one point from each $R_v$, and
thus points of $T$ form a committee of size $O(k)$ with minimax score 
$\lceil (m/k) \log k \rceil$.

To reduce the committee size to $k$ while increasing the score by only a constant factor,
we enlarge each ball $R_v$ to include the $\alpha (m / k) \log k$ closest candidates of $v$, for an appropriate constant $\alpha$.
Each ball is now $\epsilon'$-heavy, for $\epsilon' = \alpha \log k/k$, which guarantees 
an $\epsilon'$-net $T \subseteq \mathcal{C}$ for $(\mathcal{C},\mathcal{R})$ of size $O(\frac{1}{\epsilon'} \log \frac{1}{\epsilon'}) = O(k/\alpha)$.
With an appropriate choice of $\alpha$, we can ensure $|T| \leq k$ and achieve the score of
$\alpha (m / k) \log k = O((m/k) \log k)$.

When $d \leq 2$, the $\epsilon$-nets of this set system have size $O(\frac{1}{\epsilon})$, and
therefore we can construct a committee of size $k$ with score $O(m/k)$.
\end{proof}

The $O((m/k) \log k)$ and $O(m/k)$ bounds of Theorem~\ref{thm:MCC-upperbound} are essentially the best possible.
In particular, we can construct instances of Euclidean elections in which no size-$k$ 
committee can achieve the minimax score better than $\Omega(m/k)$. 
\begin{theorem}
\label{thm:MCC-lowerbound}
For any $d \geq 1$, there exist Euclidean elections in $\mathbb{R}^d$ such that any committee $T \subseteq \mathcal{C}$ of size $k$ has score $\Omega(m/k)$, where $m = |\mathcal{C}|$.
\end{theorem}
\begin{proof}
We give an example of a one-dimensional election instance where any size-$k$ committee has score $\Omega(m/k)$ which can then be embedded in higher dimensions.

Consider a set of $n$ voters on the real line at positions in $[n]$, and a set of $m$ candidates at positions $i (n /m)$, for $i \in [m]$, assuming $n$ is a multiple of $m$.
An optimal size-$k$ committee corresponds to a partition of the line into $k$ pieces, each containing $\Theta (n/k)$ voters.
Each such group also contains $\Theta(m/k)$ candidates, and only one of them is in the committee. 
Therefore, in each group, there is at least one voter (e.g., the leftmost or the rightmost) whose score is $\Omega(m /k)$.
\end{proof}

\subsection[case-c-subseteq-v]{Lower Bound on the Optimal Score when $\mathcal{C} \subseteq V$}
We now consider a special class of election instances \emph{when the candidate set is a subset of the voter set} (namely, $\mathcal{C} \subseteq V$).
We note that most representative elections satisfy this condition because each candidate is also a voter.

Notice that Theorem~\ref{thm:MCC-lowerbound} cannot be used to bound the approximation ratio of our algorithm in Theorem~\ref{thm:MCC-upperbound}, because the lower bound is only derived for the hard instances constructed in the proof.
In what follows, we prove that there is a lower bound of $\Omega(m/k)$ on the optimal score for all the instances where $\mathcal{C} \subseteq V$.

\begin{theorem}
\label{thm-lowerbound2}
For a $d$-Euclidean election $E = (\mathcal{C},V)$ with $\mathcal{C} \subseteq V$, any size-$k$ committee in $E$ has score $\Omega(m/k)$, where $m = |\mathcal{C}|$.
\end{theorem}
\begin{proof}
We will only show the proof in $d = 2$ dimensions, since a similar idea works for all constant dimensions $d \geq 3$.
Consider a sub-election $E' = (\mathcal{C}', V')$ of $E$ such that $\mathcal{C}' = V' = \mathcal{C}$, i.e., $E'$ only contains those points in $E$ which have both a candidate and a voter placed on them.
Hence, the number of candidates and the number of voters in $E'$ is $m$.
We will show that any size-$k$ committee in $E'$ has score $\Omega(m/k)$.
Since $V' \subseteq V$ and $\mathcal{C}' = \mathcal{C}$, it implies that the score of any size-$k$ committee in $E$ is also $\Omega(m/k)$.
Therefore, for the rest of this proof, we will only work with election $E'$.

\begin{figure}[t!]
    \centering
    \includegraphics[scale=0.45]{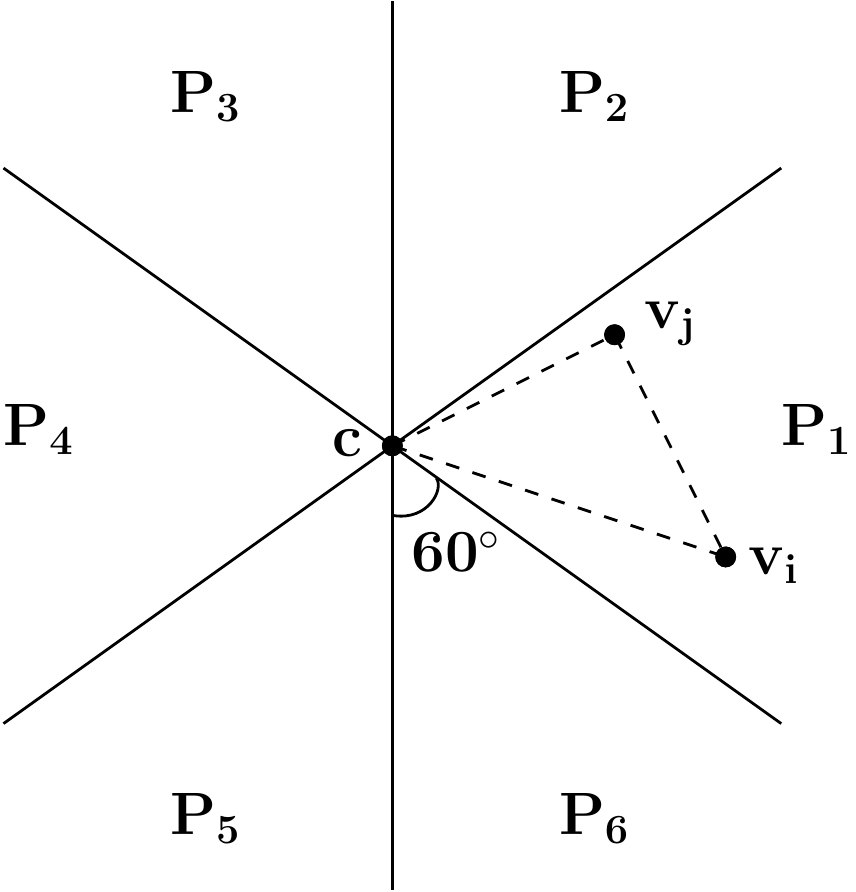}
    \caption{Partition of space across the candidate $c$. 
    Adjacent solid lines form a $60^{\circ}$ angle at $c$. 
    The part $P_1$ contains the voters $v_i, v_j$.}
    \label{fig:coincide}
\end{figure}

We begin by showing the following structural lemma for the constructed election $E'$.
\begin{lemma}
\label{lem:rank-lower-bound}
In $E' = (\mathcal{C}', V')$, each candidate belongs to the set of the closest $s$ candidates for at most $6(s-1)+1$ voters.
In other words, a candidate can be one of the top $s$ choices for at most $6(s-1)+1$ voters.
\end{lemma}
\begin{proof}
The proof is trivial for $s = 1$; hence, we consider the case when $1 < s \leq m$.
Consider an arbitrary candidate $c \in \mathcal{C}'$.
We equipartition the space into six parts $P_1, P_2, \ldots, P_6$ using three lines across $c$ (see Figure~\ref{fig:coincide}).
We assume that no candidate or voter lies on any of these three lines (note that this can be ensured by slightly moving the candidates/voters while ensuring that the rankings for each voter does not change).

We claim that $c$ belongs to the set of the closest $s$ candidates for at most $s$ voters in each part $P_i$ for $i \in [6]$.
Without loss of generality, we prove our claim for $P_1$.
Let $\hat{P} = \{\hat{p}_1, \hat{p}_2, \ldots, \hat{p}_q\}$ where $q \geq s$ (when $q \leq s$, the proof is trivial) be the points in $P_1$ which have a voter and a candidate.
Note that $\hat{P}$ is sorted according to the distance of it's points from $c$.
Let $\hat{V} = \{\hat{v}_1, \hat{v}_2, \ldots, \hat{v}_q\}$ and $\hat{C} = \{\hat{c}_1, \hat{c}_2, \ldots, \hat{c}_q\}$ be the set of voters and candidates, respectively, such that $\hat{v}_i, \hat{c}_i$ are located at $\hat{p}_i$.
We will show that for each pair $i,j \in [q]$ with $j < i$, $v_i$ prefers $c_j$ to $c$.
We need to show that $\text{dist}(v_i, c_j) < \text{dist}(v_i, c)$.
Recall that $\text{dist}(v_i, c) > \text{dist}(v_j,c)$.
Hence, using the sine rule, we know $\angle cv_jv_i > \angle cv_iv_j$.
Furthermore, if $\text{dist}(v_i, c_j) > \text{dist}(v_i, c)$, then using the sine rule, we get $\angle v_icv_j > \angle cv_jv_i > \angle cv_iv_j$.
But observe that $\angle v_icv_j < 60^\circ$.
This implies that $\angle v_icv_j + \angle cv_jv_i + \angle cv_iv_j < 180^\circ$, which is a contradiction.
Hence, $\text{dist}(v_i, c_j) < \text{dist}(v_i, c)$.
This completes the proof of Lemma~\ref{lem:rank-lower-bound}.
\end{proof}

Let $T = \{t_1, t_2, \ldots, t_k\}$ be an optimal committee in $E'$.
We partition the set of voters $V'$ into $k$ parts $V_1', V_2',\ldots, V_k' $ such that voters in $V_i'$ have $t_i$ as their most preferred candidate in $T$ for $i\in [k]$.
Using an averaging argument, we know that there exists some index $j \in [k]$ such that $V_j'$ contains at least $m/k$ voters.
Since all the voters in $V_j'$ are represented by $t_j$, using Lemma~6, we conclude that there is a voter $v \in V_j'$ such that $\sigma_v(T) \geq \frac{m}{6(k-1)+1}$.
This completes the proof of Theorem~\ref{thm-lowerbound2}.
\end{proof} 

In light of the above lower bound, it follows that
whenever the candidate set is a subset of the voter set, Theorem~\ref{thm:MCC-upperbound} implies an $O(1)$-approximation of the minimax score for $d=2$, and an $O(\log k)$-approximation for $d \geq 3$.

\begin{corollary}
Given a $d$-Euclidean election $E = (\mathcal{C}, V)$ with $\mathcal{C} \subseteq V$, we can compute in polynomial time a size-$k$ committee with minimax score within a constant factor of the optimal for $d=2$, and within a factor of $O(\log k)$ for $d \geq 3$.
\end{corollary}

\subsection[r-borda]{Minimax Committees for the $r$-Borda Rule}
A natural generalization of the Chamberlin Courant rule is the so-called \textit{$r$-Borda rule} where the score of each voter is determined by its nearest $r$ candidates in the committee for a given $r \leq k$.
More specifically, the score of a voter $v$ with respect to a committee is the sum of the ranks of its nearest $r$ candidates in the committee in the preference list of $v$.
The minimax score of a committee $T$ is the maximum over all voter scores.
Our goal is to find a committee $T$ of size $k$ that minimizes $\sigma(T)$, where
$$ \sigma(T) = \max\limits_{v \in V} \left(\min\limits_{Q \subseteq T, |Q| = r} \left(\sum\limits_{c \in Q} \sigma_v(c)\right)\right).$$

We show that for any election $E = (\mathcal{C},V)$ in $\mathbb{R}^d$, we can compute a size-$k$ committee with an $r$-Borda  score of $O((r^2m/k) \log k)$ using a modification of the algorithm in Theorem~\ref{thm:MCC-upperbound}.

\begin{theorem}
\label{thm:r-Borda-upperbound}
Given an election in any fixed dimension $d$, we can find in polynomial time a size-$k$ committee with minimax $r$-Borda score $O((r^2m/k) \log k)$.
Furthermore, if $d \leq 2$, the score can be further improved to $O(r^2m/k)$.
\end{theorem}
\begin{proof}
We only give a high-level idea of our algorithm as the rest of the details are similar to the proof of Theorem~\ref{thm:MCC-upperbound}.

For each voter $v \in V$, we create a ball $R_v$ centered at $v$ that contains $\lceil \alpha (rm /k) \log k \rceil + r$ candidates for a sufficiently large constant $\alpha$.
Let $\mathcal{R} = \{R_v: v \in V\}$ and $\epsilon = \alpha (r/k) \log k$.
Our algorithm runs in $r$ rounds.
In each round, we first compute an $\epsilon$-net $T_0 \subseteq \mathcal{C}$ for $(\mathcal{C},\mathcal{R})$ of size $O(\frac{1}{\epsilon} \log \frac{1}{\epsilon})$.
Since $\alpha$ is sufficiently large, we have $|T_0| \leq k/r$.
Then we add the candidates in $T_0$ to the committee $T$ and remove them from $\mathcal{C}$.
After $r$ rounds, we obtain our committee $T$, which is of size at most $k$.
To see that the $r$-Borda score of $T$ is $O((r^2m/k) \log k)$, we observe that each ball $R_v$ contains at least $r$ candidates in $T$, which implies that the score of $v$ with respect to $T$ is $r(\lceil \alpha (rm /k) \log k \rceil + r) = O((r^2m/k) \log k)$.
Recall that $T$ is the (disjoint) union of $r$ $\epsilon$-nets.
If $R_v$ contains at least one candidate in each of the $r$ $\epsilon$-nets, then $|R_v \cap T| \geq r$, and we are done.
So suppose $R_v$ does not contain any candidate in the $\epsilon$-net generated in the $i$-th round.
This means $R_v$ is not $\epsilon$-heavy in the $i$-th round.
But $R_v$ contains at least $m/\epsilon + r$ candidates in the original $\mathcal{C}$.
Therefore, in the $i$-th round, at least $r$ candidates in $R_v$ are removed from $\mathcal{C}$, and they are already included in $T$.
It follows that $|R_v \cap T| \geq r$.

In $\mathbb{R}^2$, the score of the committee can be further improved to $O(r^2m/k)$ using the same approach with $\epsilon$-nets of size $O(1/\epsilon)$ for disks.
This completes the proof of Theorem~\ref{thm:r-Borda-upperbound}.
\end{proof}

We observe that the bound $O(r^2m/k)$ in the above theorem is tight.
In particular, there are instances for which an optimal committee's $r$-Borda score 
is $\Omega(r^2m/k)$---this can be verified using the instance described in the proof of Theorem~\ref{thm:MCC-lowerbound}---and so this serves as the benchmark score for $r$-Borda.

\section{Approximation by Relaxing the Committee Size}
\label{sec:approx-committee-size}
The hardness result of Theorem~\ref{thm-hardtoappx} rules out any efficient algorithm for 
Euclidean election with a good approximation guarantee \emph{but only under the rigid 
constraint that the committee size is at most $k$}.
In this section, we study the problem in the setting where we are allowed to relax the committee 
size.
Specifically, we ask the following natural question: 
{\sl Can we efficiently compute a committee of size slightly larger than $k$ whose score is close to the optimal score of a size-$k$ committee?}

First, we show that if we are allowed to increase the committee size by a (small) \emph{multiplicative} factor, 
then one can achieve (or improve) the optimal score of a size-$k$ committee.

\begin{theorem} \label{thm-(1+eps)k}
Given a $d$-Euclidean election, we can compute in polynomial time a committee of size $(1+\epsilon) k$, given any fixed $\epsilon > 0$, when $d=2$ and of size $O(k \log m)$, where $m$ is the number of candidates, when $d\geq 3$,
whose score is smaller than or equal to the score of any size-$k$ committee.
\end{theorem}
\begin{proof}
We prove the result for a Euclidean election $E = (\mathcal{C},V)$ in $d=2$; 
the proof for higher dimensions is similar.
Suppose we know that the optimal size-$k$ committee of $E$ has score $\sigma^\star$.
We show how to compute a committee of size $(1+\epsilon) k$ whose score is at most $\sigma^\star$.
For each voter $v \in V$, we consider the smallest disk $R_v$ centered at $v$ containing its closest  $\sigma^\star$ candidates in $\mathcal{C}$.
Let $\mathcal{R} = \{R_v: v \in V\}$.
A \textit{hitting set} of the set system $(\mathcal{C},\mathcal{R})$ is a subset $H \subseteq \mathcal{C}$ such that $H \cap R \neq \emptyset$ for all $R \in \mathcal{R}$.
If a committee has score at most $\sigma^\star$, then it must be a hitting set of $(\mathcal{C},\mathcal{R})$, and since there is a size-$k$ committee with score $\sigma^\star$, the minimum hitting set of $(\mathcal{C},\mathcal{R})$ has size at most $k$.
By using the PTAS for disk hitting set \cite{mustafa2010improved}, we can compute a hitting set for $(\mathcal{C},\mathcal{R})$ of size $(1+\epsilon) k$, which is the desired committee.
Since we do not know the value of $\sigma^\star$, we simply try all values from $1$ to $m$, and pick the smallest one for which we have a hitting set of size $(1+\epsilon) k$.

In higher dimensions, we can apply the same approach.
The only difference is that we do not have a PTAS for ball hitting set in $\mathbb{R}^d$ for $d \geq 3$.
But we can apply the greedy hitting set algorithm to compute a hitting set of size $O(k \log m)$ if a size-$k$ hitting set exists.
Therefore, the above algorithm computes a committee of size $O(k \log m)$ whose score is at most the score of any size-$k$ committee.
\end{proof}

On the other hand, we prove that if we can only increase the committee size by an \emph{additive} constant, we are not able to achieve any good approximation for the minimax score.

\begin{theorem} \label{thm-k+c-inapprox}
Let $\alpha,\epsilon>0$ be  constants.
Given a Euclidean election $E = (\mathcal{C},V)$ in $\mathbb{R}^d$ for $d \geq 2$ and a number $k \geq 1$, it is \NPH{} to compute a committee of size $k+\alpha$ whose score is at most $|\mathcal{C}|^{1-\epsilon} \cdot \sigma^\star$, where $\sigma^\star$ is the minimum score of a size-$k$ committee.
\end{theorem}
\begin{proof}
Our proof is a minor modification of the proof of Theorem~\ref{thm-hardtoappx}; hence, we describe the construction and only give the main idea of the proof of equivalence.

We begin by constructing an election instance $E' = (\mathcal{C}', V')$ as described in the proof of Theorem~\ref{thm-hardtoappx} and do the further modifications as follows:
Let $E = (\mathcal{C},V)$ be an election instance constructed by taking $\alpha+1$ distinct copies of $E'$, i.e., $\mathcal{C} = \mathcal{C}_1' \cup \mathcal{C}_2' \cup \cdots \cup \mathcal{C}_{\alpha+1}'$ and $V = V_1' \cup V_2' \cup \cdots \cup V_{\alpha+1}'$.
In the Euclidean embedding of the election $E$, we keep the Euclidean embeddings of the individual copies of $E'$ the same as before but place these $\alpha+1$ copies far away from each other so that any voter $v \in V_i'$ prefers all candidates in its copy to any candidate in any other copy.
Recall that in $E'$, the committee size is $k'=N + n -3m$ according to the construction in Theorem~\ref{thm-hardtoappx} (where $N$, $n$, and $m$ are the total number of pieces, variable, and clauses, in the orthogonal embedding of the PM-3SAT instance, respectively).
For the election $E$ we set the committee to $k = (\alpha+1)k'$. 
This completes the construction of the reduced instance.

The main idea here is that any committee $T$ for the reduced election instance $E$ can be viewed as the disjoint union of candidates in the committees ($T_i'$) for each individual election $E_i'$.
This is because a voter $v \in V_i'$ prefers any candidate $c_i \in \mathcal{C}_i'$ to a candidate $c_j \in \mathcal{C}_j'$ for all $i, j \in [\alpha]$ with $i \neq j$.
Hence, even if we select a committee of size at most $k+\alpha$ in $E$ (i.e., $|T| \leq k+\alpha$), at least one of the copies of $E'$ (in $E$) will have a committee of size $k'$ (i.e. there exists $T_i'$ such that $|T_i'| \leq k'$).
Without loss of generality, assume $E_1'$ is one such copy.
Using Theorem~2, we know that it is \NPH{} to achieve even a $|\mathcal{C}_1'|^{1-\epsilon}$-approximation for Euclidean minimax committee in $E_1'$ for any $\epsilon > 0$.
Since we consider minimax Chamberlin-Courant rule, the approximation factor in $E$ is the max over the approximation factors in all copies of $E'$ in $E$.
This implies that a polynomial-time algorithm cannot achieve a $|\mathcal{C}_1'|^{1-\epsilon}$-approximation for any $\epsilon > 0$ in $E$.
Since $|\mathcal{C}| \leq (\alpha+1)|\mathcal{C}_1'|$, where $\alpha$ is a constant, this completes the proof of Theorem~\ref{thm-k+c-inapprox}.
\end{proof}

\section[delta-optimal]{$\delta$-Optimal Committees}
\label{sec:delta-optimal-committees}

In the previous section, we showed that for any instance we can find a minimax committee of optimal score if we increase the committee size by a small (multiplicative) factor. 
In this section, we suggest an alternative way to assess the approximation quality
\emph{while keeping the committee size $k$.}

To introduce this criterion, let us consider an election $E = (\mathcal{C},V)$ and suppose the optimal score of a size-$k$ committee is $\sigma^\star$.
In our approximation, we are looking for a size-$k$ committee in which the candidate closest to each voter $v \in V$ has rank not much larger than $\sigma^\star$ in the preference list of $v$.
Our hardness proof shows that in general this is not possible because there may be many candidates at roughly the same distance from $v$, \emph{but with a large difference in ranks}, and any polynomial time algorithm is bound to end up with a bad minimax score for some $v$.
A natural way to get rid of this pathological situation is to treat
two candidates with roughly the same distance from $v$ as if they have similar ranks.

With this motivation, we introduce the following \textit{$\delta$-optimality criterion}.
We say a committee $T \subseteq \mathcal{C}$ of size $k$ is \emph{$\delta$-optimal}, for $\delta \geq 1$, if for each voter $v \in V$ the distance from $v$ to its closest candidate in $T$ is at most $\delta$ times the distance from $v$ to its rank-$\sigma^\star$ candidate.

We now show how to compute a $3$-optimal committee in polynomial time for any $d$-Euclidean election.
Let $E = (\mathcal{C},V)$ be a Euclidean election and $k \geq 1$ be the desired committee size.
For convenience, let us first assume that the optimal score $\sigma^\star$ of a size-$k$ committee of $E$ is known.
For each voter $v \in V$, define $d_v^\star$ as the distance from $v$ to the rank-$\sigma^\star$ candidate in the preference list of $v$.
We say a voter $v$ is \textit{satisfied} with a subset $T \subseteq \mathcal{C}$ if there exists a candidate $c \in T$ such that $\text{dist}(c,v) \leq 3 d_v^\star$.
We denote by $S[T] \subseteq V$ the subset of voters satisfied with $T$.
Then a committee $T \subseteq \mathcal{C}$ (of size $k$) is $3$-optimal if every $v \in V$ is satisfied with $T$.
Our algorithm begins with an empty committee $T = \emptyset$ and iteratively adds new candidates to $T$ using the following three steps until $S[T] = V$:
\begin{enumerate}
    \item $\hat{v} \leftarrow \arg \min_{v \in V \backslash S[T]} d_v^\star$.
    \item $\hat{c} \leftarrow$ a candidate within distance $d_{\hat{v}}^\star$ from $\hat{v}$.
    \item $T \leftarrow T \cup \{\hat{c}\}$.
\end{enumerate}
In words, in each iteration, we find the \emph{unsatisfied} voter $\hat{v}$ with the minimum $d_{\hat{v}}^\star$, and then add to $T$ a (arbitrarily chosen) candidate $\hat{c} \in \mathcal{C}$ within distance $d_{\hat{v}}^\star$ from $\hat{v}$.
The algorithm terminates when $S[T] = V$, and so all voters are satisfied with $T$ at the end.

We only need to show that $|T| \leq k$; if $|T| < k$, we can always add extra candidates 
while keeping all voters satisfied.
Consider an optimal size-$k$ committee $T_\text{opt} = \{c_1,\dots,c_k\}$, with score $\sigma^\star$, and let $V_i \subseteq V$ be the subset of voters whose closest candidate in $T_\text{opt}$ is $c_i$. Thus, $V_1,\dots,V_k$ form a partition of $V$.
We say two voters $v,v' \in V$ are \textit{separated} if they belong to different $V_i$'s.

Suppose our algorithm terminates in $r$ iterations. We will show that $r \leq k$.
Let $\hat{v}_j$ (resp., $\hat{c}_j$) be the voter $\hat{v}$ (resp., the candidate $\hat{c}$) chosen in the $j$-th iteration, and let $T_j$ be the committee $T$ at the beginning of the $j$-th iteration. We claim the following property of our greedy algorithm.

\begin{lemma}
The voters $\hat{v}_1,\dots,\hat{v}_{r}$ are pairwise separated.
\end{lemma}
\begin{proof}
Let $j,j' \in [r]$ such that $j \neq j'$, and we want to show that $\hat{v}_j$ and $\hat{v}_{j'}$ are separated.
Without loss of generality, assume $j<j'$ and
let $i \in [k]$ be such that $\hat{v}_j \in V_i$.
We claim that $V_i \subseteq S[T_{j+1}]$.
Consider a voter $v \in V_i$.
If $v \in S[T_j]$, then $v \in S[T_{j+1}]$.
So assume $v \in V_i \backslash S[T_j]$.
In this case, $d_{\hat{v}_j}^\star \leq d_v^\star$.
Since the score of $T_\text{opt}$ is $\sigma^\star$, we have $\text{dist}(c_i,v) \leq d_v^\star$ and $\text{dist}(c_i,\hat{v}_j) \leq d_{\hat{v}_j}^\star \leq d_v^\star$.
Furthermore, by the construction of $\hat{c}_j$, we have $\text{dist}(\hat{c}_j,\hat{v}_j) \leq d_{\hat{v}_j}^\star \leq d_v^\star$.
It follows that
\begin{equation*}
    \text{dist}(\hat{c}_j,v) \leq \text{dist}(\hat{c}_j,\hat{v}_j) + \text{dist}(\hat{v}_j,c_i) + \text{dist}(c_i,v) \leq 3 d_v^\star.
\end{equation*}
Therefore, $v$ is satisfied with $T_{j+1} = T_j \cup \{\hat{c}_j\}$, i.e., $v \in S[T_{j+1}]$.
Based on this, we can deduce that $\hat{v}_{j'} \notin V_i$, because $\hat{v}_{j'} \notin S[T_{j'}]$ and $V_i \subseteq S[T_{j+1}] \subseteq S[T_{j'}]$.
Since $\hat{v}_j \in V_i$ and $\hat{v}_{j'} \notin S[T_{j'}]$, $\hat{v}_j$ and $\hat{v}_{j'}$ are separated.
\end{proof}


Thus, the voters $\hat{v}_1,\dots,\hat{v}_{r}$ belong to different $V_i$'s, which implies that
$r \leq k$ and $|T| = r \leq k$, proving the correctness of our algorithm.

Finally, notice that in our algorithm we assumed $\sigma^\star$ is known, but this assumption is easy to 
get rid of.
We can try each possible value from $1$ to $m = |\mathcal{C}|$, and
choose the smallest number $\sigma^\star \in [m]$ for which the algorithm returns a
committee of size at most $k$.
Thus, we proved the following.

\begin{theorem}
\label{thm:MCC-delta-opt}
Given a Euclidean election $E = (\mathcal{C}, V)$ in any dimension and $k \geq 1$, one can compute a $3$-optimal committee of size $k$ in polynomial time.
\end{theorem}

Complementary to the above algorithmic result, we can also show the following hardness result.

\begin{theorem}
\label{thm:delta-opt-hardness}
For any $\delta < 2$, unless P=NP, there is no polynomial time algorithm to compute a $\delta$-optimal committee of size $k$ for a given Euclidean election in $\mathbb{R}^d$ for $d \geq 2$.
\end{theorem}
\begin{proof}
We will show that if there is a polynomial time algorithm which computes a $\delta$-optimal committee for $\delta < 2$, then it can be used to design a polynomial-time procedure to distinguish between satisfiable and unsatisfiable instances of PM-3SAT.

First, we reduce a PM-3SAT instance $\phi$ to a $d$-dimensional Euclidean election.
Our construction is a slight modification of the one used in Theorem~\ref{thm:MCC-NPH}. 
In particular, for each variable $x_i$, let $s_i$ be the piece adjacent to the reference point $\hat{x_i}$ placed above it.
We modify $s_i$ as follows: First, we remove all points (along with the candidates/voters placed on them) from $s_i$.
Next, we move $\hat{x_i}$ up to a distance $0.5$ above $s^-_i$ so that $\hat{x_i}$ lies at the midpoint of $s_i$, i.e., the $x$-axis now passes through the midpoint of $s_i$.
Furthermore, we change the variable gadget to only have three points: one at the reference point $\hat{x_i}$, and one above and one below it at a distance of $1/6$.
We put a candidate and a voter at each of these three points.
The total number of candidates/voters in the variable gadgets is $3n$.
Moreover, for each clause gadget $z_i$, we place only a voter at the clause reference point $\hat{z_i}$.
Overall, the clause gadgets contain $m$ voters (and no candidates).
Finally, we change each nonempty piece gadget $s$ to only contain three points at distances $1/6, 1/2, 5/6$ from $s^-$, and we place a candidate and a voter at each of these three points.
The total number of candidates/voters in the piece gadgets is $3(N - 3m - n)$ where $N$ is the total number of pieces in the orthogonal embedding of the PM-3SAT instance.
Recall that for each of the $m$ clauses, the three pieces adjacent to the clause reference point $\hat{z_i}$ are empty, and for each variable $x_i$, there is a piece ($s_i$) which only contains points corresponding to a variable gadget.

Overall, we obtain an election $E = (\mathcal{C}, V)$ that consists of $3N - 8m$ voters and $3N - 9m $ candidates.
We now set the desired committee size to be $k = N -3m$.
Using similar arguments to the argument of equivalence for the reduction in Theorem~1, we can show that \emph{$E$ has a committee of size $k$ with score at most $3$ iff the PM-3SAT instance is satisfiable}.
Next, to show that it is unlikely to have an efficient algorithm to compute $\delta$-optimal committees for any $\delta < 2$, we now state the main structural property of our construction.

\begin{observation} 
\label{obs-top-three-choices}
In the constructed election $E = (\mathcal{C},V)$, for all voters $v \in V$, their closest three candidates are within the distance $1/3$, and the distance to the fourth ranked candidate is $2/3$.
\end{observation}

Let $\delta < 2$ be a constant and let $P$ be a polynomial-time algorithm which computes a $\delta$-optimal committee.
We will show that $P$ can distinguish satisfiable and unsatisfiable instances of PM-3SAT in polynomial time.
Consider a reduced election instance $E = (\mathcal{C},V)$ constructed from a PM-3SAT instance $\phi$, and let $T$ be the $\delta$-optimal committee returned by $P$.
For a committee $T$, we compute the distance $d_{max}$ which is the maximum distance of a voter $v \in V$ from its most preferred candidate in $T$, i.e. $d_{max} = \max_{v \in V} (\min_{c \in T} \text{dist}(v,c))$.
Clearly, given a committee $T$, $d_{max}$ can be computed in a polynomial time.

\begin{lemma} 
\label{lem-dv<2/3}
The instance $\phi$ is satisfiable iff $d_v < 2/3$.
\end{lemma}
\begin{proof}
We first show that if $\phi$ is a satisfiable PM-3SAT instance then $d_{max} < 2/3$.
Let $E = (\mathcal{C},V)$ be the election constructed using $\phi$.
Recall that if $\phi$ is satisfiable, then $\sigma^* \leq 3$ in the reduced instance $E$ (i.e., the optimal rank in $E$ is at most 3).
We now use Observation~\ref{obs-top-three-choices} to conclude that $d_v^* \leq 1/3$ for all voters $v \in V$.
Hence, in a $\delta$-optimal committee, each voter has its representative within distance strictly less than $2/3$, i.e., $d_{max} < 2/3$.
We now turn to the other direction and show that if $d_{max} < 2/3$ then $\phi$ is satisfiable.
We use Observation~\ref{obs-top-three-choices} to conclude that each voter $v \in V$ is represented by one of its closest three candidates, i.e., $\sigma^* \leq 3$.
Therefore, $\phi$ is satisfiable.
This completes the proof for Lemma~\ref{lem-dv<2/3}.
\end{proof}

Hence, algorithm $P$ combined with our reduction gives a polynomial-time procedure to distinguish satisfiable and unsatisfiable instances of PM-3SAT which is impossible, unless P=NP.
This completes the proof of Theorem~\ref{thm:delta-opt-hardness}.
\end{proof}

\section{Concluding Remarks}

We studied the multiwinner elections in Euclidean space under minimax Chamberlin-Courant voting rules. 
First, we settle the complexity of the problem by showing it is \NPH{} for dimensions $d\geq 2$ (in contrast, the problem is known to be polynomial time solvable in 1d), but our main contribution is presenting several attractive (\& nearly-optimal) approximation bounds which are elusive in the general setting.
We believe that our algorithms are robust and will generalize to many other interesting questions
, for instance, most of our algorithmic results (except for Thm. 5) extend to the recently studied egalitarian $k$-median rules \cite{guptaeven} when considered for the Euclidean elections. 

Our work suggests many natural directions, including resolving the complexity and approximation bounds for other important voting rules, such as the utilitarian variant of the CC rule and the general class of OWA rules \cite{skowron2015achieving,skowron2016finding} for the Euclidean elections.
Investigating the parameterized complexity of the problems we considered is also an interesting question since natural parameters such as the committee size or the number of candidates are small in many real-world elections.
Another interesting avenue is to explore Euclidean elections when the positions of voters and candidates are only approximately known, for instance, each placed at some unknown point in a disk.
Such data uncertainty naturally exists in many real-world applications.


\bibliographystyle{abbrvnat} 
\bibliography{references}

\clearpage
\begin{appendices}

\section{Proof of Lemma~\ref{lem:equivalence-if-part}}
\label{app:proof-lemma-1}
We recall the statement of Lemma~1:

``The score of $~T$ in the election $E$ is at most $4$.''

\begin{proof}[Proof (of Lemma~1)]
For the reduced Euclidean minimax committee instance $E = (\mathcal{C},V)$ and the constructed committee $T \subseteq \mathcal{C}$, we will show that, for each voter $v \in V$, $\sigma_v(T) \leq 4$.

Let us begin with the voters in the variable gadgets.
For each variable $x_i$, recall that we place four voters (say $v_1, v_2, v_3, v_4$) and four candidates (say $c_1, c_2, c_3, c_4$) near the reference point $\hat{x}_i$.
Among the four voters (resp., candidates), $v_1, v_2$ (resp., $c_1, c_2$) are at distances $0.01$ and $0.02$ above $\hat{x}_i$, and $v_3, v_4$ (resp. $c_3, c_4$) are at distances $0.01$ and $0.02$ below $\hat{x}_i$, respectively.
We claim that for each of $v_1, v_2, v_3, v_4$, the closest four candidates to that voter belong to the set $\{c_1, c_2, c_3, c_4\}$.
This is because for each of these four voters, any candidate that belongs to its nearest piece gadget is at least at a distance $0.47$ from that voter, which is strictly more than the distance from that voter to $c_j$ for all $j\in[4]$, which is bounded by $0.04$.
Now recall that our constructed committee $T$ either includes the topmost ($c_2$) or the bottommost ($c_4$) of the four candidates in the $x_i$-gadget; hence, the score for each of $v_1, v_2, v_3$, and  $v_4$ is at most $4$.
This completes the argument for the voters in all the variable gadgets.

Next, we consider the voters in the piece gadgets.
Let $s$ be a piece and let $v_1,v_2,v_3$, and $v_4$ (resp., $c_1, c_2, c_3$, and $c_4$) be the voters (resp., candidates) at distances $0.49, 0.8, 0.9, 1$ from $s^-$, respectively.
Furthermore, let $x_i$ be the variable associated with $s$. 
We only show the proof when $s$ lies above $\hat{x}_i$ (the case when $s$ lies below $\hat{x}_i$ is similar).
We first consider the case when $\pi(x_i) = \mathsf{true}$.
In this case, $T$ includes $c_4$.
For each of the voters, $v_2, v_3$ and $v_4$, their distance from $c_4$ is at most $0.2$.
On the other hand, for each of the voters $v_2, v_3$ and $v_4$, their distance from a candidate in an adjacent piece gadget is at least at $0.49$ (and, their distance from $c_1$ is at least $0.31$).
Hence, we conclude that the closest three candidates for $v_2, v_3$ and $v_4$ belong to the set $\{c_2,c_3,c_4\}$.
Therefore, the score for each these voters is at most $3$.
Finally, for the voter $v_1$, $\text{dist}(v_1, c_4) = 0.51$.
Hence, $v_1$'s closest four candidates are $\{c_1, c_2, c_3, c_4'\}$ where $c_4'$ is either the candidate belonging to the preceding piece gadget of $s$, placed at $s^-$ or it is the candidate placed at a distance $0.02$ above the variable reference point $\hat{x}_i$ (we are in the latter case when $s$ is adjacent to $\hat{x}_i$).
Therefore, due to the way we construct $T$, in either case, $T$ includes $c_4'$.
(Note that $c_4'$ exists even when $s^-$ is a connection point.)
Hence, the score of $v_1$ is at most 4.
We now consider the case when $\pi(x_i) = \mathsf{false}$.
In this case, $T$ includes the candidate $c_3$ from $s$.
As argued above, $c_3$ belongs to the set of closest four candidates for each of the four voters $v_1, v_2, v_3, v_4$.
Hence, the score of committee $T$ for each of the four voters is at most $4$.

Finally, we consider the voters from the clause gadget.
We only show our argument for an arbitrary positive clause $z_i$, as the proof for the negative clauses is similar.
We need to show that at least one of the four closest candidates to $z_i$ belongs to $T$.
Consider the three pieces $s_1,s_2,s_3$ adjacent to the clause reference point $\hat{z}_i$.
Suppose $s_1$ is to the left of $\hat{z}_i$, $s_2$ is to the right of $\hat{z}_i$, and $s_3$ is below $\hat{z}_i$.
By our construction, we have $\hat{z}_i = s_1^+ = s_2^+ = s_3^+$.
Since $\hat{z}_i$ is a connection point and all connection points have even coordinates, $s_1^-,s_2^-,s_3^-$ are not connection points.
Therefore, there exist pieces $s_4,s_5,s_6$ such that the right endpoint of $s_4$ is $s_1^-$, the left endpoint of $s_5$ is $s_2^-$, and the top endpoint of $s_6$ is $s_3^-$.
We have $s_1^- = s_4^+$, $s_2^- = s_5^+$, and $s_3^- = s_6^+$.
In the $s_4$-gadget, there is a candidate $c_4$ with distance $1$ from $s_4^-$ (and hence located at $s_1^-$).
Similarly, there is a candidate $c_5$ in the $s_5$-gadget located at $s_2^-$ and a candidate $c_6$ in the $s_6$-gadget located at $s_3^-$.
The candidates $c_4,c_5,c_6$, together with the candidate in the $z_i$-gadget (which is located at $\hat{z}_i$), are the four candidates closest to the voter at $\hat{z}_i$, because all pieces except $s_1,\dots,s_6$ have distances at least $2$ from $\hat{z}_i$ by the fact that the $x$-coordinates (resp., $y$-coordinates) of the vertical (resp., horizontal) pieces are all even numbers.
We claim that $T$ includes at least one of $c_4, c_5$ or $c_6$.
Recall that $\pi$ is a valid satisfying assignment.
Hence, at least one of the associated variables of $s_4, s_5$ or $s_6$ is $\mathsf{true}$ under $\pi$.
Therefore, by the construction of $T$, for at least one of $s_4, s_5$ or $s_6$, $T$ contains the candidate belonging to it which is placed at distance $1$ from $s_4^-, s_5^-$ or $s_6^-$, respectively (i.e., $T$ contains one of $c_4, c_5$ or $c_6$).
This completes the proof of Lemma~1.
\end{proof}

\newpage
\section{Proof of Lemma~\ref{lem-hard-appx-score}}
\label{app:proof-lemma-4}
We recall the statement of Lemma~4:

``If $\phi$ is satisfiable, then there exists a size $k$ committee with score at most 4;
otherwise, every size $k$ committee has score at least $(nm)^w$."

\begin{proof}[Proof (of Lemma~4)]
First, to show that when $\phi$ is satisfiable, there exists a size-$k$ committee with score at most $4$, we refer the reader to the ``if'' part of the argument of equivalence for the \NPH{}ness result in Theorem~1 where we construct a committee $T$. 
It can be easily verified that the same committee $T$ also has the score at most $4$ in the instance constructed for Lemma~5.

Next, we show that if $\phi$ is unsatisfiable, then every size-$k$ committee has score at least $(nm)^w$.
For a contradiction, assume that when $\phi$ is unsatisfiable, then in the reduced instance, there is a size-k committee $T$ with score strictly less than $(nm)^w$.
We will first show the following claim.

\begin{claim} \label{clm-one-cand}
$T$ contains exactly one candidate in each variable gadget and exactly one candidate in each piece gadget.
Moreover, $T$ contains no candidate in the clause gadgets.
\end{claim}
\begin{proof}
We first show that $T$ contains at least one candidate in each variable gadget.

For each variable $x_i$, recall that we place four voters (say $v_1, v_2, v_3, v_4$) and four candidates (say $c_1, c_2, c_3, c_4$) near the reference point $\hat{x}_i$.
Among the four voters (resp., candidates), $v_1, v_2$ (resp., $c_1, c_2$) are at distances $0.01$ and $0.02$ above $\hat{x}_i$, and $v_3, v_4$ (resp. $c_3, c_4$) are at distances $0.01$ and $0.02$ below $\hat{x}_i$, respectively.
Moreover, let $q^+$ (resp., $q^-$) be the point to the left of $v_2$ (resp., $v_4$) at a distance $0.05$; we place $(nm)^w$ candidates at $q^+$ (resp., $q^-$).
Observe that the closest $(nm)^w$ candidates for the voters $v_1, v_2, v_3, v_4$ belong to the $x_i$-gadget.
This is because any candidate from the nearest piece gadget is at least at a distance $0.4$.
Hence, $T$ contains at least one candidate from each variable gadget.
Next, we consider a piece gadget $s$.
We first recall the construction of $s$.
Let $s_1, s_2, s_3$, and $s_4$ be the points placed at distances $0.47, 0.8, 0.9, 1$ from $s^-$, and let $s_3', s_4'$ be the points placed at a distance $0.19$ to the left of the points $s_3, s_4$, respectively.
We place a voter and a candidate at each of $s_1, s_2, s_3$ and $s_4$, and place $(nm)^w$ candidates at $s_3'$ and $s_4'$.
Notice that the closest $(nm)^w$ candidates for the voter placed at $s_4$ include the candidates placed $s_4, s_3$ and $s_4'$ (which belong to $s$).
Hence, $T$ contains at least one candidate from each piece gadget.

At this stage, we recall that the number of variable gadgets is $n$, the number of piece gadgets is $N - 3m$, and the committee size is $k = N+n-3m$.
Hence, using a counting argument, we conclude that $T$ contains exactly one candidate from each variable gadget and each piece gadget.
Clearly, $T$ cannot contain any candidate from the clause gadgets.
This completes the proof of Claim~1.
\end{proof}

Next, we strengthen Claim~1 as follows:
For each piece gadget $s$, we show that $T$ contains the candidate placed at either $s_3$ or $s_4$. 
Indeed, consider the voters placed at $s_3$ and $s_4$.
The closest $(nm)^w$ candidates for the voter at $s_3$ include the candidates placed at $s_1, s_2, s_3, s_4, s_3'$.
And, for the voter at $s_4$, it's closest $(nm)^w$ candidates include the candidates placed at $s_3, s_4, s_4'$.
By Claim~1, we know $T$ includes exactly one candidate in $s$; hence, $T$ must include a candidate placed at either $s_3$ or $s_4$ to bound the score for the voters at $s_3$ and $s_4$ simultaneously.
By a similar argument, for each variable gadget, we can show that $T$ must include $c_1, c_2, c_3$ or $c_4$ (in particular, $T$ cannot include any candidates placed at $q^+$ or $q^-$).
Therefore, we have shown that $T$ does not contain any candidate placed on one of the points with $(nm)^w$ candidates.

We now claim that we can recover a satisfying assignment $\pi$ for the PM-3SAT instance $\phi$.
We define $\pi$ as follows: For each variable $x_i$, if $T$ includes a candidate from the $x_i$-gadget placed above the $x$-axis, then $\pi(x_i) = \textsf{true}$; otherwise $\pi(x_i) = \textsf{false}$.
We can now use a similar argument to the ``only if'' part of the argument of equivalence for the \NPH{}ness result from Theorem~1 to show that $\pi$ is a satisfying assignment of $\phi$ (due to the similarity of the arguments, we skip the details).
But this is a contradiction, since we started with an unsatisfiable instance $\phi$ of PM-3SAT.
This completes the proof of Lemma~4.
\end{proof}

\end{appendices}

\end{document}